\documentclass{ifacconf}

\usepackage{graphicx}      
\usepackage{natbib}        
\usepackage{bm}
\usepackage{array}
\usepackage{amsmath,amssymb,amsfonts}
\usepackage{comment}
\usepackage[dvipsnames]{xcolor}
\def\RevAll#1{\textcolor{black}{#1}}
\begin{document}
\begin{frontmatter}

\title{Event-triggered parameter estimator \\for sensor fusion}

\thanks[footnoteinfo]{This work was supported in part by the Secretaría de Ciencia, Humanidades, Tecnología e Innovación (SECIHTI), México, previously administered by the Consejo Nacional de Humanidades, Ciencias y Tecnologías (CONAHCYT) with grant number 1229622, and in part by projects PID2021-124137OB-I00 and 
PID2024-159279OB-I00 funded by MICIU/AEI/10.13039/501100011033 and by ERDF/EU, by  
project REMAIN S1/1.1/E0111 (Interreg Sudoe Programme, ERDF), and via project DGA T45\_23R (Gobierno de Aragón). Grant reference BG24/00121 funded by MICIU/AEI/10.13039/501100011033. {\color{red} \copyright 2026 IFAC. This work has been accepted to IFAC for publication under a Creative Commons Licence CC-BY-NC-ND. Accepted for presentation at the 23rd IFAC World Congress 2026.}}

\author[First,Second]{Ariana Méndez-Castillo} 
\author[Second]{Irene Perez-Salesa}
\author[Second]{Rodrigo Aldana-López} 
\author[First]{Antonio Ramírez-Treviño}
\author[Second]{Rosario Aragues}

\address[First]{Department of Electrical Engineering, Cinvestav-Guadalajara, Jalisco, 45019 México (e-mail: ariana.mendez@cinvestav.mx, antonio.ramirezt@cinvestav.mx)}
\address[Second]{Department of Computing and Systems Engineering - I3A, University of Zaragoza, Zaragoza 50009 España (e-mail: i.perez@unizar.es, raldana@unizar.es, raragues@unizar.es ) }

\begin{abstract}                
This paper studies event-triggered parameter estimation in sensor fusion systems where sensors transmit measurements to a gradient based estimator. We introduce a regressor-driven local triggering rule that requires no knowledge of the current parameter estimate and depends solely on the regressor signals. Under a persistent excitation condition on the aggregate regressor, we derive explicit design inequalities on the estimator gain and event thresholds that guarantee global exponential convergence. The analysis is based on a time-varying Lyapunov function. We further provide a sufficient condition on the regressor dynamics that enforces a uniform lower bound on inter-event times, excluding Zeno behavior. Simulations show substantial communication savings while preserving exponential convergence.
\end{abstract}


\end{frontmatter}

\section{Introduction}
Parameter estimation is a data regression problem appearing in adaptive control, sensor networks, and unmanned surface vehicles \citep{Shen2024}. Precise knowledge of the unknown parameters is required to guarantee stability and closed loop performance in several applications \citep{olfati2007consensus}. In this article, we address the problem of estimating unknown parameters from measurements collected by a network of sensors, while aiming to reduce the overall use of the communication network.

Several estimation algorithms exist in the literature. A standard approach relies on modeling the unknown parameters as constant states of an underlying dynamical system. Therefore, one may apply a discrete-time Kalman filter \citep{kalman1960new, maybeck1982stochastic} or a continuous-time Kalman–Bucy filter \citep{golovan1994kalman}. In both cases, the covariance of the estimation error evolves according to a Riccati equation, which guarantees asymptotic convergence of the estimation error. To obtain exponential convergence, different estimators have been introduced, most notably gradient-based schemes \citep{ioannou1996robust, sastry2011adaptive}. For this class of estimators, exponential convergence of the parameter estimates to their true values is guaranteed under the standard persistent excitation (PE) condition \citep{anderson1997exponential}. However, classical parameter estimation methods are primarily formulated for settings in which sensors and estimators can interact instantaneously.

In contrast, sensor networks, where sensing units and processing units are often spatially distributed, arise in many modern applications. In such networks, the information available at any individual sensor is often insufficient to estimate the unknown parameters of interest. As a result, the measurements collected across the network are transmitted to an estimator and processed collectively, a strategy commonly referred to as sensor fusion \citep{SASIADEK2002203}. This approach is a core component in the design of autonomous vehicles and related systems \citep{Kocic2018}.

When parameter estimation relies on a sensor network, reducing communication and computational load becomes essential. Event-triggering mechanisms (ETMs) are commonly used for this purpose, since they restrict transmissions to times when a triggering rule is met. Several classes of ETMs have been introduced in the literature, mostly in contexts such as state estimation or control. A standard example is the send on delta mechanism \citep{Miskowicz2014}, where a measurement is transmitted only when it deviates from the last transmitted value by more than a fixed threshold. This rule depends solely on local measurements. In \citep{DIAO201895}, a send on delta strategy combined with recursive stochastic approximation yields asymptotic convergence of the estimation error, but the analysis is restricted to binary valued measurements. Other ETMs include stochastic triggering rules \citep{Shen2024}, where transmission occurs when the deviation between successive measurements exceeds a random threshold. This setup considers joint state and parameter estimation by augmenting the parameter vector into the state and applying a Kalman filter, which guarantees only asymptotic convergence. Dynamic ETMs, such as those in \citep{Perez-Salesa2025}, further reduce transmissions by comparing previously received information with a dynamic threshold and exploiting negative information between events. Event-triggered parameter estimators have also been proposed \citep{Geng2024}, typically using consensus plus innovation schemes with adaptive triggers based on prediction errors. These rules require each sensor to evaluate the predicted output, which in turn requires real-time access to the current parameter estimate. In general, these methods are designed for state estimation and rely on a Kalman filter structure, and thus achieve asymptotic stability. 

Motivated by this discussion, the contributions of this work are as follows.
\begin{itemize}
    \item We propose an event-triggered gradient estimator for sensor fusion, where each sensor uses a local trigger based on changes in its regressor and the fusion unit runs a continuous–time gradient update with piecewise constant aggregated data.

    \item We derive explicit conditions on the estimator gain and local thresholds that ensure global exponential convergence under persistent excitation of the aggregate regressor, using a time–varying Lyapunov function analysis.

    \item We provide a sufficient condition on the regressor dynamics that rules out Zeno behavior.
    
    \item We present numerical experiments showing that exponential convergence is preserved with a substantial reduction in sensor transmissions.
\end{itemize}

\section{Problem Statement}

We consider a sensor network with $N$ agents. Each agent $i \in \{1,\dots,N\}$ collects measurements
$$
\mathbf{y}_i(t) = \mathbf{C}_i(t)\, \bm{\theta},
$$
where $\mathbf{y}_i(t) \in \mathbb{R}^{p_i}$ is the measurement vector, $\mathbf{C}_i(t) \in \mathbb{R}^{p_i \times n}$ is a time-varying regressor, and $\bm{\theta} \in \mathbb{R}^n$ is the unknown parameter vector of interest.

Local estimation of $\bm{\theta}$ is often infeasible because an individual regressor may not have enough richness or satisfy persistent excitation conditions. Moreover, individual agents may have limited computational capabilities to process all collective information, particularly for large $N$. Therefore, agents forward information to an aggregator node, which computes a global estimate of $\bm{\theta}$ using all data received up to time $t \ge 0$.

To reduce network load, we adopt event-triggered communication so that transmissions occur only when a prescribed triggering condition is met. The goal is to design an estimator together with triggering mechanisms that ensure exponential convergence towards the origin for the estimation error, under standard persistent excitation assumptions on the aggregate regressor across all sensors. The overall architecture is illustrated in Fig.~\ref{fig:system:model}.
\begin{figure}
    \centering
    \includegraphics[width=0.8\linewidth]{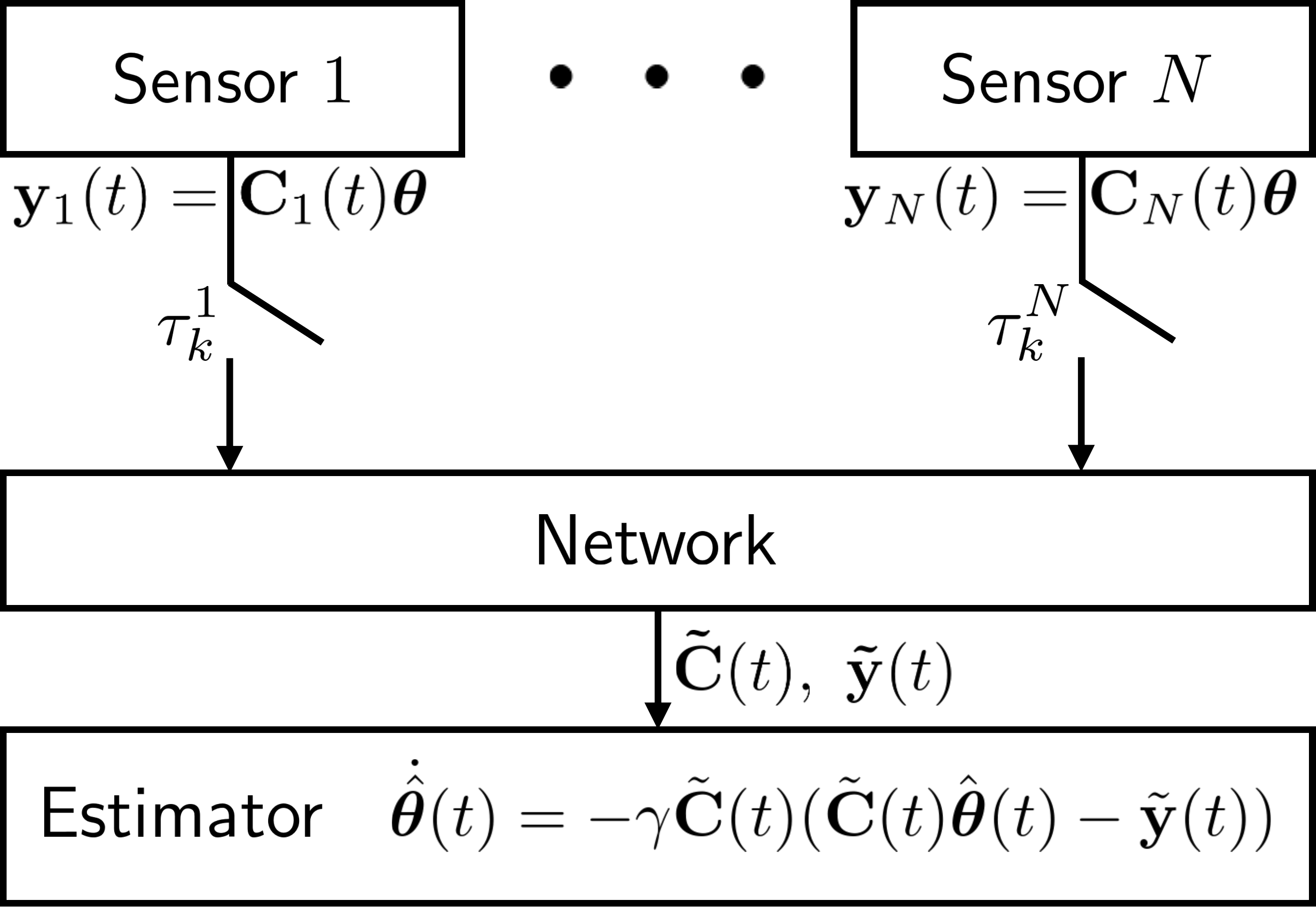}
    \caption{Parameter estimator using event-triggered sensor fusion. Spatially distributed sensors collect measurements and transmit them only when a local triggering rule is satisfied. An estimator node receives the transmitted data and computes the estimate of the unknown parameters using all aggregated measurements.}
    \label{fig:system:model}
\end{figure}

\section{Main results}

We propose the following local event-triggering mechanism and  estimator algorithm as follows:
\begin{subequations}\label{eq:event_trigger_estimator}
\begin{align}
    &\textbf{Local event-triggering mechanism:} \nonumber\\
    &\tau_{k+1}^i 
      = \sup\Big\{ t \ge \tau_k^i : \nonumber\\
    &\hspace{5em}
    \big\| \mathbf{C}_i(\tau_k^i)^\top \mathbf{C}_i(\tau_k^i)
        - \mathbf{C}_i(t)^\top \mathbf{C}_i(t) \big\|
        \ge \overline{\Delta}_i \Big\}
      \label{eq:local_trigger}
    \\
    &\textbf{Information from agent $i$ used by the estimator:} \nonumber\\
    &\!\begin{array}{lcll}
    \tilde{\mathbf{C}}_i(t) &=&
     \mathbf{C}_i(\tau_k^i), &\text{ for } t \in [\tau_k^i,\, \tau_{k+1}^i)
      \\[0.2em]
    \tilde{\mathbf{y}}_i(t) &=&
        \mathbf{y}_i(\tau_k^i) &\text{ for } t \in [\tau_k^i,\, \tau_{k+1}^i)
    \end{array}
      \label{eq:received_info}
    \\
    &\textbf{Estimator:} \nonumber\\
    &\dot{\hat{\bm{\theta}}}(t)
      = -\, {\gamma}\,
        \tilde{\mathbf{C}}(t)^\top
        \big( \tilde{\mathbf{C}}(t)\hat{\bm{\theta}}(t)
        - \tilde{\mathbf{y}}(t) \big)
      \label{eq:estimator}
\end{align}
\end{subequations}
where $\tau_k^i$ denotes the $k$-th triggering instant of sensor $i$, at which the sensor transmits 
$\mathbf{y}_i(\tau_k^i)$ and $\mathbf{C}_i(\tau_k^i)$ to the aggregator node. Consequently, 
$\tilde{\mathbf{y}}_i(t)$ and $\tilde{\mathbf{C}}_i(t)$ represent the most recently received information from sensor $i$, and therefore remain piecewise constant between two successive triggering instants $(\tau_k^i,\, \tau_{k+1}^i)$ for each $i$. The aggregate quantities are defined as
\begin{equation}\label{eq:aggregate-evtr-measurement}
\begin{aligned}
\tilde{\mathbf{C}}(t)
&=
\begin{bmatrix}
\tilde{\mathbf{C}}_1(t)\\
\vdots\\
\tilde{\mathbf{C}}_N(t)
\end{bmatrix},
\qquad
\tilde{\mathbf{y}}(t)
=
\begin{bmatrix}
\tilde{\mathbf{y}}_1(t)\\
\vdots\\
\tilde{\mathbf{y}}_N(t)
\end{bmatrix},
\end{aligned}
\end{equation}
which collect the latest information recovered from all sensors at time $t$ and are used collectively in \eqref{eq:estimator}. Moreover, $\overline{\Delta}_i \geq 0$ and ${\gamma}>0$ are design parameters.

The event-triggering condition in \eqref{eq:local_trigger} depends only on the local exogenous regressor $\mathbf{C}_i(t)$. Hence, the presence of undesirable Zeno behavior, i.e. triggering an infinite number of events in a finite time interval, depends only on the properties of this exogenous input. For the sake of generality in the main results, we will assume the absence of Zeno behavior. However, note that, in most practical applications, it is excluded due to the digital handling of signals, which forces a minimum inter-event time. In any case, in the following we characterize a practical condition under which Zeno behavior is avoided:

\begin{prop}\label{prop:zeno}
    Assume that there exist constants $b_i,c_i>0$ such that $\|\mathbf{C}_i(t)\|\leq b_i$ and $\|\dot{\mathbf{C}}_i(t)\|\leq c_i$ for all $t\geq 0$. Then, for all $i\in\{1,\dots,N\}$, the event instants generated by \eqref{eq:local_trigger} satisfy the condition
    \begin{equation}
    \label{eq:zeno}
    \tau_{k+1}^i- \tau_k^i\geq \underline{\tau},
    \end{equation}
for some $\underline{\tau}>0$, i.e. there exists a minimum inter-event time.
\end{prop}
\begin{pf}
    Fix any $i\in\{1,\dots,N\}$. Set $\mathbf{M}_i(t)=\mathbf{C}_i(t)^\top \mathbf{C}_i(t)$. Since
$$
\dot{\mathbf{M}}_i(t)
= \dot{\mathbf{C}}_i(t)^\top \mathbf{C}_i(t)+\mathbf{C}_i(t)^\top \dot{\mathbf{C}}_i(t),
$$
we have
$
\|\dot{\mathbf{M}}_i(t)\|
\le 2\|\mathbf{C}_i(t)\|\,\|\dot{\mathbf{C}}_i(t)\|
\le 2 b_i c_i =: m_i
$ for all $t$. We proceed by contradiction. Suppose that there exists no $\underline{\tau}>0$ such that $\tau_{k+1}^i-\tau_k^i\geq \underline{\tau}$. Therefore, there exists $k\geq 0$ such that
$$
0 < \tau_{k+1}^i - \tau_k^i \le \frac{\overline{\Delta}_i}{2 m_i}
$$
where $\overline{\Delta}_i>0$ is the fixed threshold in \eqref{eq:local_trigger}. Define
$$
M_k^i
:= \left\|\frac{\mathbf{M}_i(\tau_{k+1}^i)-\mathbf{M}_i(\tau_k^i)}
{\tau_{k+1}^i-\tau_k^i}\right\|
= \frac{\overline{\Delta}_i}{\tau_{k+1}^i-\tau_k^i}\geq 2m_i,
$$
therefore,
for such $k$. On the other hand, by the mean value theorem, for every $k$ there exists  
$\xi_k^i \in [\tau_k^i,\tau_{k+1}^i]$ such that
$
M_k^i = \|\dot{\mathbf{M}}_i(\xi_k^i)\| \le m_i,
$
which contradicts $M_k^i\ge 2 m_i$. Thus, there exists a minimum inter-event time $\underline{\tau}>0$. \qed

\end{pf}

To analyze the convergence of the proposed estimator, we impose the following standard condition on the aggregate regressor in the absence of event-triggering:

\begin{assum}\label{as:pe}
Let
$$
{\mathbf{C}}(t)
=
\begin{bmatrix}
{\mathbf{C}}_1(t)\\
\vdots\\
{\mathbf{C}}_N(t)
\end{bmatrix}.
$$
There exist constants $m>0$, $b>0$, and $T>0$ such that, for all $t\ge 0$,
\begin{equation}
m\mathbf{I}_n \preceq \int_{t-T}^t \mathbf{C}(\tau)^\top \mathbf{C}(\tau)\, \mathrm{d}\tau,
\qquad
\|\mathbf{C}(t)\| \le b.
\end{equation}
\end{assum}

The main result of this work is stated next. It establishes sufficient conditions on the design parameters $\overline{\Delta}_i$ and $\gamma$ that guarantee exponential convergence of the estimation error to the origin.

\begin{thm}\label{th:trigger}
Let Assumption~\ref{as:pe} and the Proposition \ref{prop:zeno} hold. Define
\begin{equation}\label{eq:bound_const}
\begin{aligned}
\beta &= \frac{\gamma T^2 b^3}{2},\\
k_1 &= \frac{1}{2\gamma}\!\left(\frac{2\beta^2}{m} + T\right),\\
k_2&=k_1+\frac{T^2 b^2}{2}
\end{aligned}
\end{equation}
and impose the design conditions
\RevAll{\begin{equation}\label{eq:design}
\gamma>0,
\qquad
\overline{\Delta}_i \le \frac{m}{4\gamma k_2N}.
\end{equation}}
Then, the origin for the estimation error $\tilde{\bm{\theta}}(t) := \hat{\bm{\theta}}(t) - \bm{\theta}$ is globally exponentially stable under the dynamics defined in
\eqref{eq:event_trigger_estimator}.
\end{thm}

\section{Convergence analysis}

To prove Theorem~\ref{th:trigger}, we introduce several technical lemmas that support the convergence analysis. We begin by rewriting the estimation error dynamics in a form that isolates the effect of the event-triggering mechanism.

\begin{lem}\label{lem:error_dynamics}
Let the conditions of Theorem~\ref{th:trigger} hold. Then the estimation error 
$\tilde{\bm{\theta}}(t) = \hat{\bm{\theta}}(t) - \bm{\theta}$ evolves according to
\begin{equation}\label{eq:error:system}
\dot{\tilde{\bm{\theta}}}(t)
= -\gamma\, \mathbf{C}(t)^\top \mathbf{C}(t)\tilde{\bm{\theta}}(t)
  + \gamma\, \bm{\Delta}(t)\tilde{\bm{\theta}}(t),
\end{equation}
where $\mathbf{C}(t)$ is the aggregate regressor defined in Assumption~\ref{as:pe}, and
\begin{equation}\label{eq:Delta}
\bm{\Delta}(t):= \mathbf{C}(t)^\top \mathbf{C}(t)- \tilde{\mathbf{C}}(t)^\top \tilde{\mathbf{C}}(t),
\end{equation}
with $\tilde{\mathbf{C}}(t)$ defined in \eqref{eq:aggregate-evtr-measurement} according to the received sensor information.
\end{lem}

\begin{pf}
Given that the value of the matrices $\tilde{\mathbf{C}}(t)$ and $\tilde{\mathbf{y}}(t)$ from \eqref{eq:aggregate-evtr-measurement} is only updated at event instants, its value remains constant between the occurrence of two events. For the analysis, we define the errors induced by the event-triggered communication as $\bm{\varepsilon}_c (t) = \mathbf{\tilde{C}} (t) - \mathbf{C}(t)$ and $\bm{\varepsilon}_y (t) = \tilde{\mathbf{y}} (t) - \mathbf{y}(t)$ in the interval $t \in (\tau_k, \ \tau_{k+1})$. 
Substituting in the estimator \eqref{eq:estimator},
$$\dot{\hat{\bm{\theta}}}(t)\!
      =\! -\, {\gamma}\,
        \big(\bm{\varepsilon}_c (t) + \mathbf{C}(t)\big)\!^\top\!
        \big( (\bm{\varepsilon}_c (t) + \mathbf{C}(t))\hat{\bm{\theta}}(t)
        - (\bm{\varepsilon}_y (t) + \mathbf{y}(t)) \big).
        $$
        
Considering $\mathbf{y}(t)=\mathbf{C}(t) \bm{\theta}$, $\bm{\varepsilon}_y(t) = \bm{\varepsilon}_c(t)\bm{\theta}$, we have
\begin{align*}
    \dot{\bm{\hat{\theta}}}(t) =& -\gamma \mathbf{C}(t)^\top \mathbf{C}(t) (\bm{\hat{\theta}}(t)-\bm{\theta}) \\
    &- \gamma\Big(\mathbf{C}(t)^\top \bm{\varepsilon}_c(t) \!+\! \bm{\varepsilon}_c(t)^\top \big( \bm{\varepsilon}_c(t) + \mathbf{C}(t)\big) \Big) (\bm{\hat{\theta}}(t)\!-\!\bm{\theta}).
\end{align*}
Defining the estimation error as $\bm{\tilde{\theta}}(t) = \bm{\hat{\theta}}(t) - \bm{\theta}$,
\begin{align*}
    &\dot{\bm{\tilde{\theta}}}(t) = -\gamma \mathbf{C}(t)^\top \mathbf{C}(t) \bm{\tilde{\theta}}(t) \\
    & - \gamma\Big(\mathbf{C}(t)^\top \big(\mathbf{\tilde{C}} (t) - \mathbf{C}(t) \big) + \big(\mathbf{\tilde{C}} (t) - \mathbf{C}(t) \big)^\top \mathbf{\tilde{C}}(t) \Big) \bm{\tilde{\theta}}(t)\\
    &= -\gamma \mathbf{C}(t)^\top\mathbf{C}(t) \bm{\tilde{\theta}}(t)+\gamma \big( \mathbf{C}(t)^\top \mathbf{C}(t) -\mathbf{\tilde{C}}(t)^\top \mathbf{\tilde{C}}(t)\big) \bm{\tilde{\theta}}(t).
\end{align*}
Then, which coincides with \eqref{eq:error:system}, completing the proof. \qed
\end{pf}

Using the expression for $\bm{\Delta}(t)$ in \eqref{eq:Delta}, we next characterize its uniform boundedness imposed by the event-triggering mechanism \eqref{eq:local_trigger}.
\RevAll{\begin{lem}\label{lem:delta_bound}
Let the conditions of Theorem~\ref{th:trigger} hold. Then the disturbance term \eqref{eq:Delta} satisfies
$$\|\bm{\Delta}(t)\| \le N \overline{\Delta},
\qquad \forall\, t \ge 0$$
with $\overline{\Delta}=\max_i \{\overline{\Delta}_i\}$ being the maximum event threshold from \eqref{eq:local_trigger}.
\end{lem}}
\RevAll{\begin{pf}
Note that $\bm{\Delta}(t)$ adds up all $\bm{\Delta}_i(t) = \mathbf{C}_i(t)^\top \mathbf{C}_i(t) - \mathbf{\tilde{C}}_i(t)^\top \mathbf{\tilde{C}}_i(t)$, in the form $\bm{\Delta}(t)=\sum_{i=1}^N \bm{\Delta}_i(t)$. Every $\bm{\Delta}_i(t)$ is bounded as $\|\bm{\Delta}_i (t) \| \leq \overline{\Delta}_i \leq \overline{\Delta}$, as a consequence of the event-triggering mechanism \eqref{eq:local_trigger}. Then, we have
\begin{align*}
    \|\bm{\Delta}(t) \| & = \|\sum_{i=1}^N \bm{\Delta}_i(t)\| \leq \sum_{i=1}^N\| \bm{\Delta}_i(t)\| \leq N \overline{\Delta}
\end{align*} 
completing the proof. \qed
\end{pf}}
Following the methodology presented in \citep{rueda2021strong}, adapted to the event-triggered multi-sensor setting considered here, we consider the Lyapunov function candidate
\begin{equation}\label{eq:lyap}
V(\tilde{\bm{\theta}},t)
= \tilde{\bm{\theta}}^\top \mathbf{P}(t)\tilde{\bm{\theta}},
\end{equation}
where $\mathbf{P}(t)$ is defined as
\begin{equation}\label{eq:P}
\mathbf{P}(t)
= \frac{1}{2\gamma}\!\left(\frac{2\beta^2}{m}+T\right)\mathbf{I}_n
  + \int_{t-T}^{t} (\tau - t + T)\,\mathbf{C}(\tau)^\top\mathbf{C}(\tau)\, \text{d}\tau,
\end{equation}
where $\mathbf{I}_n$ is the identity matrix of dimensions $n\times n$. We first verify that $V$ satisfies the required quadratic bounds for it to be a proper Lyapunov function.

\begin{lem}\label{lem:Lyap_bounds}
Let the conditions of Theorem~\ref{th:trigger} hold. Then the Lyapunov function \eqref{eq:lyap} satisfies
\begin{equation}\label{eq:Lyap_bounds}
k_1 \|\tilde{\bm{\theta}}\|^2
\;\le\;
V(\tilde{\bm{\theta}},t)
\;\le\;
k_2 \|\tilde{\bm{\theta}}\|^2,
\ \ \forall\, t \ge 0, \ \ \forall \tilde{\bm{\theta}}\in\mathbb{R}^n.
\end{equation}
\end{lem}

\begin{pf}
Considering the Assumption \ref{as:pe}, and solving the integral of the second part of $\mathbf{P}(t)$ in \eqref{eq:P}:
\begin{align*}
    &\int_{t-T}^t\!(\tau\!-\!t\!+\!T) \mathbf{C}(\tau)\!^\top\! \mathbf{C}(\tau) \text{d} \tau  \preceq \! \int_{t-T}^t \! (\tau\!-\!t\!+\!T) \|\mathbf{C}(t)\|^2 \mathbf{I}_n \text{d}\tau\\
    &=\int_{t-T}^t(\tau\!-\!t\!+\!T)b^2 \mathbf{I}_n \text{d}\tau= b^2 \int_0^T s \mathbf{I}_n \text{d}s = \frac{T^2 b^2}{2}.
\end{align*}
For the lower bound, note that $0 \leq (\tau -t+T) \leq T$ for $\tau \in  [t-T,t]$. Then,
\begin{equation}\label{eq:bound_P2}
\frac{T^2 b^2}{2} \mathbf{I}_n \succeq \int_{t-T}^t(\tau-t+T) \mathbf{C}(\tau)^\top \mathbf{C}(\tau) \text{d}\tau\succeq \bm{0}.\end{equation}

Defining $k_1$ and $k_2$ as in \eqref{eq:bound_const}, $\mathbf{P}(t)$ is bounded by $k_1 \leq \|\mathbf{P}(t)\| \leq k_2$ corresponding to \eqref{eq:lyap}, completing the proof. \qed
\end{pf}

We now establish a key differential inequality relating \eqref{eq:lyap}, \eqref{eq:P}, and the nominal stabilizing term in \eqref{eq:error:system}.

\begin{lem}\label{lem:bound_g2}
Let the conditions of Theorem~\ref{th:trigger} hold. Then, for $V(\tilde{\bm{\theta}},t)$ and $\mathbf{P}(t)$ defined in 
\eqref{eq:lyap}–\eqref{eq:P}, the inequality
\begin{equation}\label{eq:bound_g2}
\tilde{\bm{\theta}}^\top \dot{\mathbf{P}}(t)\tilde{\bm{\theta}}
-
2\gamma\, \tilde{\bm{\theta}}^\top \mathbf{P}(t)\mathbf{C}(t)^\top\mathbf{C}(t)\tilde{\bm{\theta}}
\;\le\;
-\frac{m}{2}\|\tilde{\bm{\theta}}\|^2
\end{equation}
holds for all $t\ge 0$, and all $\tilde{\bm{\theta}}\in\mathbb{R}^n$.
\end{lem}

\begin{pf}
Let, $$\rho(t) := \tilde{\bm{\theta}}^\top \dot{\mathbf{P}}(t)\tilde{\bm{\theta}} - 2 \gamma \tilde{\bm{\theta}}^\top \mathbf{P}(t) \mathbf{C} (t)^\top \mathbf{C}(t) \tilde{\bm{\theta}}. $$ Differentiate $\mathbf{P}(t)$ to obtain
\begin{equation*}
\dot{\mathbf{P}}(t)
= T\, \mathbf{C}(t)^\top \mathbf{C}(t) - \int_{t-T}^{t} \mathbf{C}(\tau)^\top \mathbf{C}(\tau)\, \text{d}\tau .    
\end{equation*}
Replacing in $\rho(t)$:
\begin{align*}
    \rho (t)&= \tilde{\bm{\theta}}^\top 
   \left(T \mathbf{C}(t)^\top \mathbf{C}(t)
      - \int_{t-T}^{t}\mathbf{C}(\tau)^\top \mathbf{C}(\tau)\, \text{d}\tau\right)\tilde{\bm{\theta}}\\
   & - 2\gamma\, \tilde{\bm{\theta}}^\top \mathbf{P}(t) \mathbf{C}(t)^\top \mathbf{C}(t)\tilde{\bm{\theta}}\\
&= \tilde{\bm{\theta}}^\top 
   \bigg[-\int_{t-T}^t \mathbf{C}(\tau)^\top \mathbf{C}(\tau)\, \text{d}\tau
       \\
      & + (T\mathbf{I}_n - 2\gamma \mathbf{P}(t))\, \mathbf{C}(t)^\top \mathbf{C}(t)\bigg]
    \tilde{\bm{\theta}}.
\end{align*}
From the second term, we have
\begin{align*}
    &T\mathbf{I}_n - 2\gamma \mathbf{P}(t) =
 T\mathbf{I}_n
   - 2\gamma\!\left(
     \frac{1}{2\gamma}\!\left(\frac{2 \beta^2}{m}+T\right)\mathbf{I}_n \right.\\
     & \left. + \int_{t-T}^{t}(\tau-t+T)\, \mathbf{C}(\tau)^\top \mathbf{C}(\tau)\, \text{d}\tau \right)\\
& = -\frac{2 \beta^2}{m}\mathbf{I}_n
   - 2\gamma \int_{t-T}^t (\tau-t+T)\, \mathbf{C}(\tau)^\top \mathbf{C}(\tau)\, \text{d}\tau .
\end{align*}
Substituting the expression above in $\rho(t)$ yields
\begin{align*}
    \rho(t) &=  \tilde{\bm{\theta}}^\top 
    \bigg[-\int_{t-T}^t \mathbf{C}(\tau)^\top \mathbf{C}(\tau)\, \text{d}\tau -\bigg(\frac{2 \beta^2}{m}\mathbf{I}_n \\
   & + 2\gamma \int_{t-T}^t (\tau-t+T)\, \mathbf{C}(\tau)^\top \mathbf{C}(\tau)\, \text{d}\tau\bigg)\, \mathbf{C}(t)^\top \mathbf{C}(t)\bigg]
    \tilde{\bm{\theta}}.
\end{align*}
Using the upper bound of the integral, described in  \eqref{eq:bound_P2}, and the bound from Assumption \ref{as:pe},
\begin{align*}
    \rho(t) \leq - m \|\tilde{\bm{\theta}}\|^2
   - \frac{2 \beta^2}{m}\|\mathbf{C}(t)\tilde{\bm{\theta}}\|^2
   + 2\gamma\!\left(\frac{T^2 b^3}{2}\right)
     \|\tilde{\bm{\theta}}\|\, \|\mathbf{C}(t)\tilde{\bm{\theta}}\|.
\end{align*}
Applying Young’s inequality,
$$\|\tilde{\bm{\theta}}\|\,
\|\mathbf{C}(t)\tilde{\bm{\theta}}\|
\le \frac{\|\tilde{\bm{\theta}}\|^2}{2\varepsilon}
  + \frac{\varepsilon}{2}\|\mathbf{C}(t)\tilde{\bm{\theta}}\|^2,$$ which results in  
\begin{align*}
\rho(t)
&\le
\left(\frac{\gamma T^2 b^3}{2\varepsilon}
      - m \right)\|\tilde{\bm{\theta}}\|^2
+
\left(\frac{\gamma T^2 b^3 \varepsilon}{2}
      -\frac{2 \beta^2}{m}\right)
\|\mathbf{C}(t)\tilde{\bm{\theta}}\|^2 .
\end{align*}

Choosing $\varepsilon = \frac{2\beta}{m}$ with $\beta = \frac{\gamma T^2 b^3}{2}$, we reach $\rho(t) \leq -\frac{m}{2}\|\tilde{\bm{\theta}}\|^2,$ completing the proof. \qed
\end{pf}

\subsection{Proof of Theorem ~\ref{th:trigger}}
First note that condition \eqref{th:trigger} ensures that the trajectory $\bm{\theta}(t)$ from \eqref{eq:estimator} exists for all $t\geq 0$. Lemma~\ref{lem:error_dynamics} shows that the estimation error 
$\tilde{\bm{\theta}}(t) = \hat{\bm{\theta}}(t) - \bm{\theta}$ satisfies the perturbed dynamics
\eqref{eq:error:system}, where the perturbation $\bm{\Delta}(t)$ is defined in
\eqref{eq:Delta}. Consider now the Lyapunov function $V(\tilde{\bm{\theta}},t)$ defined in \eqref{eq:lyap},
with $\mathbf{P}(t)$ given by \eqref{eq:P}. Applying the product rule gives
$$
\dot{V}(\tilde{\bm{\theta}}(t),t)
= \tilde{\bm{\theta}}(t)^\top \dot{\mathbf{P}}(t)\tilde{\bm{\theta}}(t)
  + 2\,\tilde{\bm{\theta}}(t)^\top \mathbf{P}(t)\dot{\tilde{\bm{\theta}}}(t).
$$
Using \eqref{eq:error:system} in the second term yields
$$
\begin{aligned}
\dot{V}(\tilde{\bm{\theta}(t)},t)
&=  \tilde{\bm{\theta}}(t)^\top \dot{\mathbf{P}}(t)\tilde{\bm{\theta}}(t)
  - 2\gamma\, \tilde{\bm{\theta}}(t)^\top \mathbf{P}(t)\mathbf{C}(t)^\top\mathbf{C}(t)\tilde{\bm{\theta}}(t)\\&
  + 2\gamma\, \tilde{\bm{\theta}}(t)^\top \mathbf{P}(t)\bm{\Delta}(t)\tilde{\bm{\theta}}(t).
  \end{aligned}
$$
The first two terms exactly match the combination appearing in Lemma~\ref{lem:bound_g2}, implying
$$
\dot{V}(\tilde{\bm{\theta}}(t),t)
\le
-\frac{m}{2}\|\tilde{\bm{\theta}}(t)\|^2
+ 2\gamma\, \tilde{\bm{\theta}}(t)^\top \mathbf{P}(t)\bm{\Delta}(t)\tilde{\bm{\theta}}(t).
$$

Using the bound in Lemma~\ref{lem:delta_bound} together with the upper bound in
\eqref{eq:Lyap_bounds} from Lemma \ref{lem:Lyap_bounds}, we obtain
 \RevAll{$$
\begin{aligned}
2\gamma \tilde{\bm{\theta}}(t)^\top \mathbf{P}(t)\bm{\Delta}(t)\tilde{\bm{\theta}}(t)
&\le 2\gamma\, \|\mathbf{P}(t)\|\|\bm{\Delta}(t)\|\|\tilde{\bm{\theta}}(t)\|^2\\&
\le 2\gamma k_2 N \overline{\Delta}\|\tilde{\bm{\theta}}(t)\|^2.
\end{aligned}
$$}
Hence,
\RevAll{$$
\dot{V}(\tilde{\bm{\theta}}(t),t)
\le
-\left(\frac{m}{2} - 2\gamma k_2 N \overline{\Delta}\right)\|\tilde{\bm{\theta}}(t)\|^2 = -k_2\alpha \|\tilde{\bm{\theta}}(t)\|^2,
$$}
with 
\RevAll{$$
\alpha = \frac{1}{k_2}\left(\frac{m}{2} - 2\gamma k_2 N\overline{\Delta}\right),
$$}
where condition $\alpha>0$ is ensured by the design condition on $\overline{\Delta}_i$ in Theorem~\ref{th:trigger}. Moreover, applying Lemma \ref{lem:Lyap_bounds} again leads to: 
$$
\dot{V}(\tilde{\bm{\theta}}(t),t)
\le -\alpha\, V(\tilde{\bm{\theta}}(t),t),
\qquad \forall\, t\ge 0.
$$
By the standard comparison lemma and a final application of Lemma \ref{lem:Lyap_bounds}, this implies
$$
k_1\|\tilde{\bm{\theta}}(t)\|^2\leq V(\tilde{\bm{\theta}}(t),t)
\le
V(\tilde{\bm{\theta}}(0),0)\exp(-\alpha t),
\quad \forall\, t\ge 0,
$$
obtaining global exponential decay of the estimation error
$\tilde{\bm{\theta}}(t)$ towards the origin. Hence, the equilibrium
$\tilde{\bm{\theta}}=\bm{0}$ of the estimation error dynamics is globally exponentially stable, completing the proof.

\section{Numerical example}
As an illustrative example, consider a network with $N=3$ sensors and one estimator. 
The unknown parameter vector $\bm{\theta} \in \mathbb{R}^n$ with $n=5$ is drawn from a uniform distribution on $[-10,10]^n$. 
For each sensor $i$, the regressor matrix $\mathbf{C}_i(t) \in \mathbb{R}^{p_i \times n}$, with $p_i=2$, has entries of the form 
$A + B \sin(\omega t) + C \cos(\omega t)$, where $\omega, A, B, C$ are sampled from uniform distributions over $[0,3]$. 
From simulations, we obtain $T = 5$, $b = 6.82$, and $m = 19$, which satisfy Assumption~\ref{as:pe}. 
For simplicity, we set $\overline{\Delta}_i=\overline{\Delta}$ for all sensors in the experiments where
\[
\overline{\Delta} := \max_{i \in \{1,\dots,N\}} \overline{\Delta}_i.
\]
Furthermore, we set $\gamma=0.1$.

The resulting trajectories for this scenario are shown in Figure~\ref{fig:example}. 
Each component of the estimation vector $\hat{\bm{\theta}}(t)$ approaches the true parameter vector $\bm{\theta}$, and the estimation error converges to zero. 
This convergence occurs despite the estimator receiving measurements only at discrete event times, illustrated in the bottom panel of the figure, where sensors trigger at different rates according to their local dynamics. 
For this example we set $\overline{\Delta} = 2.8$. 

Next, we analyze the behavior of the algorithm as $\overline{\Delta}$ varies. 
Figure~\ref{fig:var_delta} shows that the number of triggering events decreases significantly as $\overline{\Delta}$ increases, as expected. 
Across all tested values, the estimates still converge to the true parameter vector. 
Convergence times, computed as
$$
\inf\{ t \ge 0 : \|\hat{\bm{\theta}}(t) - \bm{\theta}\| \le 0.5 \},
$$
remain finite and grow with $\overline{\Delta}$, illustrating the tradeoff between convergence speed and communication effort. While larger values of $\overline{\Delta}$ were also tested, the design condition in \eqref{eq:design} guaranteeing exponential stability requires
$
\overline{\Delta} \le 3.2 \times 10^{-4}.
$
\RevAll{$
\overline{\Delta} \le 1.012 \times 10^{-4}.
$}
This highlights the conservativeness of the theoretical bound. In practice, the estimator remains stable and convergent for much larger triggering thresholds. These observations indicate that a less conservative analysis is needed and will be pursued in future work.

\begin{figure}
    \centering
    \includegraphics[width=1 \linewidth]{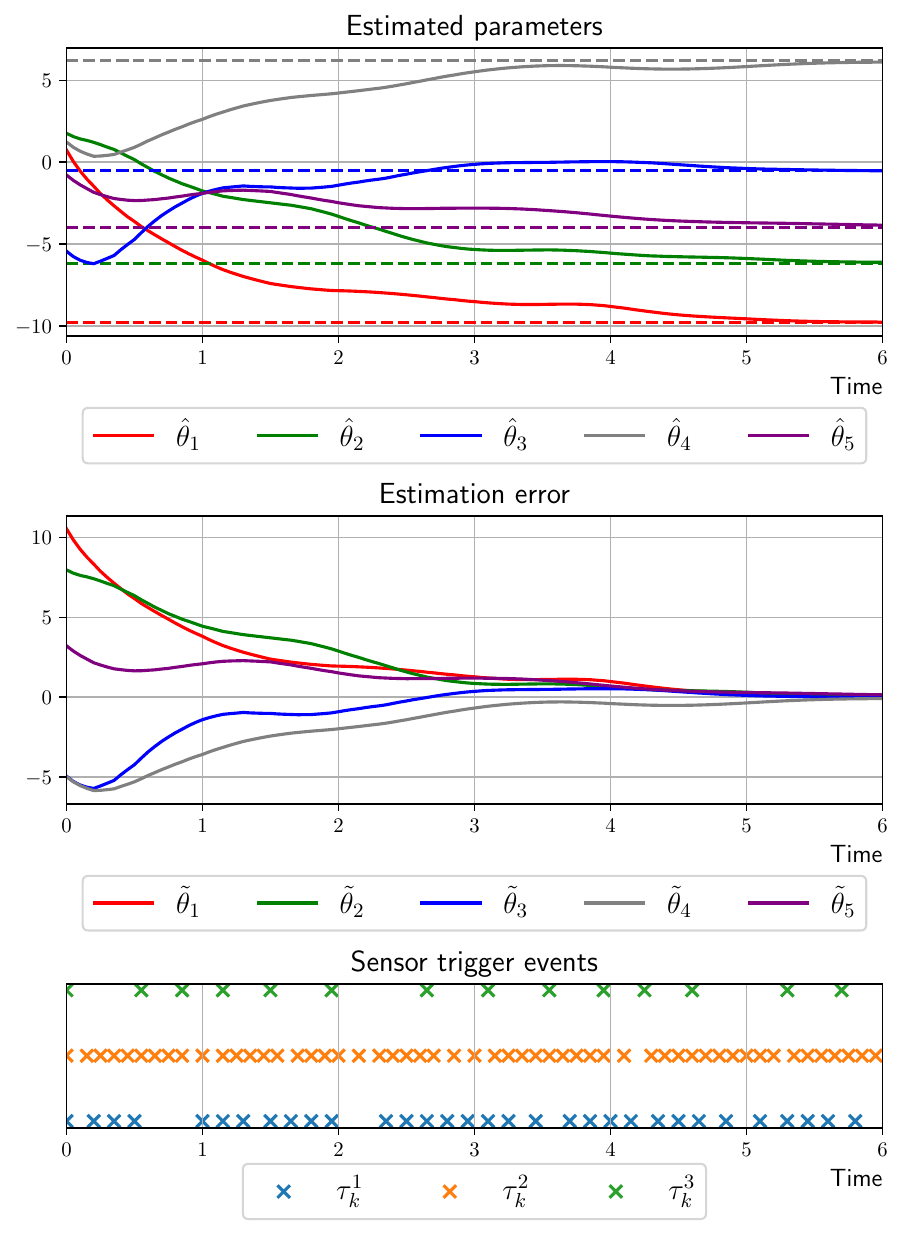}
    \caption{The estimated parameters (top) converge to the actual values, which is verified by the estimation error (middle) converging to zero. Individual event-triggering instants for each sensor are shown using $\gamma=0.1$ and $\overline{\Delta} = 2.8$ (bottom).}
    \label{fig:example}
\end{figure}

\begin{figure}
    \centering
    \includegraphics[width = 1\linewidth]{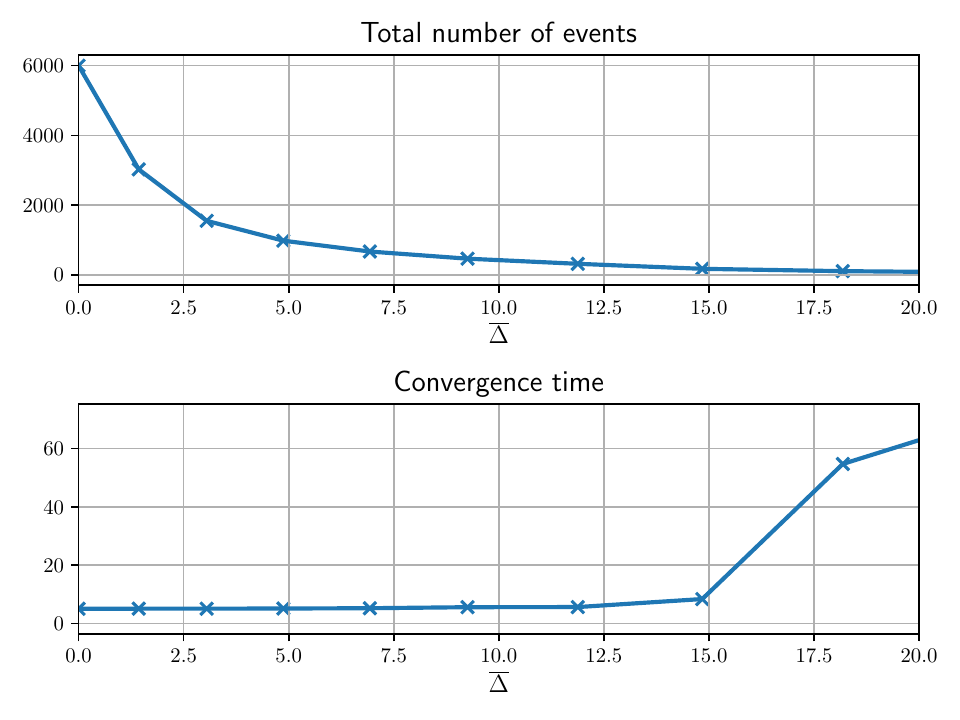}
    \caption{Total number of events (top) and convergence time (bottom) with fixed $\gamma=0.1$ and $\overline{\Delta} \in [0,20]$. The higher the threshold, the lower the total number of events sent (top) and, consequently, the longer the convergence time (bottom).}
    \label{fig:var_delta}
\end{figure}

\section{Conclusion}
This work introduced an event-triggered gradient estimator for sensor fusion, based on a regressor-driven triggering rule that does not require local evaluation of parameter estimates. We established explicit conditions on the estimator gain and triggering thresholds that guarantee global exponential convergence under persistent excitation, and  derived a sufficient condition on the regressor dynamics ensuring a minimum inter-event time. The resulting scheme removes the need for continuous communication while retaining strong stability guarantees. Simulations show a communication–performance tradeoff and demonstrated that significant reductions in transmissions can be achieved without compromising convergence. As future work, we will investigate less conservative conditions for the triggering thresholds to narrow the gap between the theoretical guarantees and the behavior observed in practice.
\bibliography{ifac2bib}
\end{document}